\documentclass[letterpaper,11pt]{article}
\usepackage[utf8]{inputenc}

\pdfoutput=1

\usepackage{amsmath, amsthm, amssymb, array}
\usepackage[linesnumbered,vlined]{algorithm2e}
\usepackage{mathtools}
\usepackage[sort&compress,numbers]{natbib}
\usepackage{comment} 
\usepackage{color-edits} 
\usepackage[most]{tcolorbox}
\usepackage{xfrac}
\usepackage{hyperref}
\usepackage{multirow}
\usepackage{caption}
\usepackage{bm}
\usepackage{newfloat}
\usepackage{enumitem}
\usepackage{xpatch}
\usepackage{soul}
\usepackage[T1]{fontenc}

\usepackage[margin=1in]{geometry}

\allowdisplaybreaks

\definecolor{mygreen}{RGB}{20,120,60}

\title{Fully Dynamic Matching: \\ Beating 2-Approximation in $\Delta^\epsilon$ Update Time\footnote{A preliminary version of this paper is to appear in proceedings of SODA 2020.}}

\author{
Soheil Behnezhad\footnote{Supported in part by a Google PhD Fellowship. Research conducted in part while the author was a research intern at Google.}\\University of Maryland \and 
Jakub Łącki\\ Google Research \and
Vahab Mirrokni\\ Google Research
}

\date{}

\newcommand{\E}[0]{\ensuremath{\mathbb{E}}}

\newcommand{\opt}[0]{\ensuremath{\textsc{opt}}}

\newcommand{\lfmm}[1]{\ensuremath{\mathsf{Greedy}(#1)}}

\DeclareMathOperator{\polylog}{polylog}

\DeclareMathOperator{\elim}{elim}

\let\originalleft\left
\let\originalright\right
\renewcommand{\left}{\mathopen{}\mathclose\bgroup\originalleft}
\renewcommand{\right}{\aftergroup\egroup\originalright}

\addauthor{sb}{blue}    
\addauthor{kl}{red}    

\newtheorem{theorem}{Theorem}

\newtheorem{lemma}{Lemma}[section]

\newtheorem{claim}[lemma]{Claim}

\newtheorem{observation}[lemma]{Observation}

\renewcommand{\epsilon}{\ensuremath{\varepsilon}}

\makeatletter
\def\thm@space@setup{%
  \thm@preskip= 0.2cm
  \thm@postskip=\thm@preskip 
}
\makeatother

\definecolor{mygreen}{RGB}{20,125,20}
\definecolor{myred}{RGB}{125,20,20}
\definecolor{linkcolor}{RGB}{0,0,230}
\definecolor{quoteboxborder}{RGB}{200,200,200}
\definecolor{mylightgray}{RGB}{230,230,230}
\definecolor{verylightgray}{RGB}{230,230,230}
\definecolor{commentcolor}{RGB}{120,120,120}

\SetCommentSty{mycommfont}
\SetKwInOut{Input}{Input}
\SetKwInOut{Output}{Output}

\newcommand{\smparagraph}[1]{
\par\addvspace{0.2cm}
\noindent \textbf{#1}
}

\hypersetup{
     colorlinks=true,
     citecolor= mygreen,
     linkcolor= myred,
     urlcolor= mygreen
}

\newcommand{\etal}[0]{\textit{et al.}}

\newcounter{myalgctr}

\newenvironment{tbox}{
\par\addvspace{0.2cm}
\begin{tcolorbox}[width=\textwidth,
                  enhanced,
				  float=h,
                  boxsep=2pt,
                  left=1pt,
                  right=-12pt,
                  top=4pt,
                  boxrule=1pt,
                  arc=0pt,
                  colback=white,
                  colframe=black,
                  unbreakable
                  ]
}{
\end{tcolorbox}
}

\newenvironment{graytbox}{
\par\addvspace{0.1cm}
\begin{tcolorbox}[width=\textwidth,
                  enhanced,
                  frame hidden,
                  boxsep=5pt,
                  left=1pt,
                  right=1pt,
                  top=4pt,
                  bottom=4pt,
                  boxrule=1pt,
                  arc=0pt,
                  colback=mylightgray,
                  colframe=black,
                  breakable
                  ]
}{
\end{tcolorbox}
}

\newcommand{\tboxhrule}[0]{
\par
\vspace*{-0.3cm}
\noindent\rule{\linewidth-13pt}{0.2mm}\par
\vspace*{0.3cm}
}

\newenvironment{titledtbox}[1]{\begin{tbox} #1 \tboxhrule}{\end{tbox}}

\newenvironment{tboxalg2e}[1]{
\refstepcounter{myalgctr}
	\begin{titledtbox}{\textbf{Algorithm \themyalgctr.} #1}
	\vspace{-0.2cm}
}
{
	\vspace{-0.3cm}
	\end{titledtbox}
}

\newenvironment{quotebox}{
\par\addvspace{0.2cm}
\begin{tcolorbox}[width=\textwidth,
                  enhanced,
                  frame hidden,
                  interior hidden,
                  boxsep=0pt,
                  left=10pt,
                  right=0pt,
                  top=0pt,
                  bottom=0pt,
                  boxrule=1pt,
                  arc=0pt,
                  colback=white,
                  colframe=black,
                  borderline west={3pt}{0pt}{quoteboxborder}
                  ]
}{
\end{tcolorbox}
\smallskip
}

\newenvironment{highlighttechnical}[0]{
\vspace{0.1cm}
\begin{tcolorbox}[width=\textwidth,
                  enhanced,
                  boxsep=2pt,
                  left=1pt,
                  right=1pt,
                  top=4pt,
                  boxrule=0.8pt,
                  arc=0pt,
                  colback=verylightgray,
                  colframe=black,
                  unbreakable
                  ]
}{
\end{tcolorbox}
}

\begin{document}
\maketitle

\thispagestyle{empty}
\begin{abstract}
\setlength{\parskip}{0.4em}
In fully dynamic graphs, we know how to maintain a 2-approximation of maximum matching extremely fast, that is, in polylogarithmic update time or better. In a sharp contrast and despite extensive studies, all known algorithms that maintain a $2-\Omega(1)$ approximate matching are much slower. Understanding this gap and, in particular, determining the best possible update time for algorithms providing a better-than-2 approximate matching is a major open question.

In this paper, we show that for any constant $\epsilon > 0$, there is a randomized algorithm that with high probability maintains a $2-\Omega(1)$ approximate maximum matching of a fully-dynamic general graph in worst-case update time $O(\Delta^{\epsilon}+\polylog n)$, where $\Delta$ is the maximum degree.

Previously, the fastest fully dynamic matching algorithm providing a better-than-2 approximation had $O(m^{1/4})$ update-time [Bernstein and Stein, SODA 2016]. A faster algorithm with update-time $O(n^\epsilon)$ was known, but worked only for maintaining the size (and not the edges) of the matching in bipartite graphs [Bhattacharya, Henzinger, and Nanongkai, STOC 2016].
\end{abstract}
\clearpage

\setcounter{page}{1}
\section{Introduction}

The problem of maintaining a large matching in the dynamic setting has received significant attention over the last two decades (see  \cite{DBLP:conf/stoc/OnakR10,DBLP:journals/siamcomp/BaswanaGS18,DBLP:conf/stoc/NeimanS13,DBLP:conf/focs/GuptaP13,DBLP:journals/siamcomp/BhattacharyaHI18,DBLP:conf/stoc/BhattacharyaHN16,DBLP:conf/focs/Solomon16,DBLP:conf/soda/BhattacharyaHN17, DBLP:conf/ipco/BhattacharyaCH17,DBLP:conf/icalp/CharikarS18,DBLP:conf/icalp/ArarCCSW18,DBLP:conf/soda/BernsteinFH19} and the references therein). After a long line of work, we now know how to maintain a {\em maximal matching} in fully dynamic graphs (i.e. graphs that undergo both edge insertions and deletions) extremely fast, that is in polylogarithmic update-time or better \cite{DBLP:conf/focs/BaswanaGS11,DBLP:conf/focs/Solomon16,DBLP:conf/soda/BernsteinFH19}. This immediately gives a 2-approximation of \emph{maximum} matching. In a sharp contrast, however, we have little understanding of the update-time-complexity once we go below $2$ approximation. A famous open question of the area, asked first\footnote{The problem was also stated in multiple subsequent papers e.g. in \cite[Section~4]{BhattacharyaArxiv}, \cite[Section~7]{DBLP:conf/soda/BernsteinS16} or \cite[Section~1]{DBLP:conf/icalp/CharikarS18}.} in the pioneering paper of Onak and Rubinfeld \cite{DBLP:conf/stoc/OnakR10} from 2010 is:

\begin{quotebox}
	``{\itshape Can the approximation constant be made smaller than 2 for maximum matching {\normalfont [}while having polylogarithmic update-time}]?'' \cite{DBLP:conf/stoc/OnakR10}
\end{quotebox}

A decade later, we are still far from achieving a polylogarithmic update-time algorithm. The fastest current result for maintaining a matching with a better-than-2 approximation factor was presented by Bernstein and Stein~\cite{DBLP:conf/soda/BernsteinS16}.\footnote{The algorithm of Bernstein and Stein remarkably achieves an (almost) $1.5$ approximation.} Their algorithm handles updates in $O(m^{1/4})$ time where $m$ denotes the number of edges in the graph. However, a notable follow-up result of Bhattacharya, Henzinger, and Nanongkai~\cite{DBLP:conf/stoc/BhattacharyaHN16} hinted that we may be able to achieve a faster algorithm. They showed that in $n$-vertex bipartite graphs, for any constant $\epsilon > 0$, there is a deterministic algorithm with amortized update-time $O(n^\epsilon)$ that maintains a $2-\Omega(1)$ approximation of the size of the maximum matching (the algorithm maintains a fractional matching, but not an integral one).

In light of the result of Bhattacharya~\etal{}~\cite{DBLP:conf/stoc/BhattacharyaHN16}, two main questions remained open: First, {\em is it possible to maintain the matching in addition to its size?} Second, {\em can the result be extended from bipartite graphs to general graphs?} We resolve both questions in the affirmative:

\begin{graytbox}
	\begin{theorem}\label{thm:main}
		For any constant $\epsilon \in (0, 1)$, there is a randomized fully-dynamic algorithm that with high probability maintains a $2-\Omega_\epsilon(1)$ approximate maximum matching in worst-case update-time $O(\Delta^{\epsilon}) + \polylog n$ under the standard oblivious adversary assumption. Here, $\Delta$ denotes the maximum degree in the graph. Also the precise approximation factor constant depends on $\epsilon$.
	\end{theorem}
\end{graytbox}

Compared to the algorithm of Bhattacharya~\etal{}~\cite{DBLP:conf/stoc/BhattacharyaHN16}, our algorithm, at the expense of using randomization, maintains the matching itself,  handles general graphs, and also improves the update-time from $O(n^\epsilon)$ {\em amortized} to $\widetilde{O}(\Delta^\epsilon)$ {\em worst-case}. In addition, our algorithm is arguably simpler.

Similar to other randomized algorithms of the literature, we require the standard oblivious adversary assumption. The adversary here is all powerful and knows the algorithm, but his/her updates should be independent of the random bits used by the algorithm. Equivalently, one can assume that the sequence of updates is fixed adversarially {\em before} the algorithm starts to operate.

\section{Our Techniques}\label{sec:overview}

In this section, we provide an informal overview of the ideas used in our algorithm for Theorem~\ref{thm:main} and the challenges that arise along the way.

A main intuition behind the efficient (randomized) 2-approximate algorithms of the literature is essentially ``hiding'' the matching from the adversary through the use of randomization. For instance if we pick the edges in the matching randomly from the dense regions of the graph where we have a lot of choices, it would then take the adversary (who recall is unaware of our random bits) a lot of trials to remove a matching edge. A natural algorithm having such behavior is {\em random greedy maximal matching} (RGMM) which processes the edges in a random order and greedily adds them to the matching if possible. Indeed it was  recently shown by Behnezhad~\etal{}~\cite{misfocs} that it takes only polylogarithmic time per edge update to maintain a RGMM.

Unfortunately, exactly the feature of RGMM (or of previous algorithms  based on the same intuition) that it matches the dense regions first prevents it from obtaining a better-than-2 approximation. A simple bad example is a perfect matching whose one side induces a clique (formally, a graph on vertices $v_1, \ldots, v_{2n}$, whose edge set consists of a clique induced on $v_1, \ldots, v_n$ and edges $v_iv_{i+n}$ for $1 \leq i \leq n$). The RGMM algorithm, for instance, would pick almost all of its edges from the clique, and thus matches roughly half of the vertices while the graph has a perfect matching.


To break this 2 approximation barrier, our starting point is a slightly paraphrased variant of a streaming  algorithm of Konrad~\etal{}~\cite{streamingapprox}. The algorithm starts by constructing a RGMM $M_0$ of the input graph $G$. Unless $M_0$ is significantly larger than half of the size of the maximum matching $\opt$ of $G$, then nearly all edges of $M_0$ can be shown to belong to length-3 augmenting paths in $M_0 \oplus \opt$. Therefore to break 2 approximation, it suffices to pick a constant fraction of the edges in $M_0$, and discover a collection of vertex disjoint length-3 paths augmenting them. Konrad~\etal{}~\cite{streamingapprox} showed that this can be done by finding another RGMM, this time on a subgraph $G'$ of $G$ whose edges have one endpoint that is matched in $M_0$ and one endpoint that is unmatched (though for these edges to augment $M_0$ well, it is crucial that not all such edges are included in $G'$).

The algorithm outlined above shows how to obtain a $2-\Omega(1)$ approximate maximum matching by merely running two instances of RGMM. Given that we know how to maintain a RGMM in polylogarithmic update time due to \cite{misfocs}, one may wonder whether we can also get a similar update-time for this algorithm. Unfortunately, the answer is negative! The reason is that the second stage graph $G'$ is {\em adaptively} determined based on matching $M_0$. Particularly, a single edge update that changes matching $M_0$ may lead to deletion/insertion of a {\em vertex} (along with its edges) in the second stage graph $G'$. While we can handle edge updates in polylogarithmic time, the update-time for vertex updates is still polynomial (in the degree of the vertex being updated). Therefore, the algorithm, as stated, requires an update-time of up to $\widetilde{O}(\Delta)$.

\newcommand{\ub}[0]{\ensuremath{\alpha}}
\newcommand{\lb}[0]{\ensuremath{\beta}}

To get around the issue above, our first insight is a parametrized analysis of the update-time depending on the structure of the edges in $M_0$. Suppose that matching $M_0$ is constructed by drawing a random {\em rank} $\pi_0(e) \in [0, 1]$ independently on each edge $e$ and then iterating over the edges in the increasing order of their ranks. We show that the whole update-time (i.e. that of both the first and the second stage matchings) can be bounded by $\widetilde{O}(\rho)$, where 
\[
	\ub = \max_{e \in M_0}\pi_0(e), \qquad 
	\lb = \min_{e \in M_0}\pi_0(e), \qquad \text{ and } \qquad
	\rho = \frac{\ub}{\lb}.
\]
The reason is as follows. For an edge update $e$, the probability that it causes an update to matching $M_0$, is upper bounded by $\alpha$. (If rank of $e$ is larger than the highest rank ever in $M_0$, then $e \not\in M_0$ and thus its insertion/deletion causes no update to $M_0$.) On the other hand, in case of an update to $M_0$, the cost of a vertex update to the second stage graph can be bounded by its degree, which we show can be bounded by $\widetilde{O}(1/\beta)$ using a sparsification property of RGMM (Lemma~\ref{lem:sparsification}) applied to the first stage matching $M_0$. Thus, the overall update time is indeed $\widetilde{O}(\alpha \cdot \beta^{-1}) = \widetilde{O}(\rho)$.

The analysis highlighted above, shows that as the rank of edges in the first stage matching $M_0$ get closer to each other, our update-time gets improved. In general, this ratio can be as large as $O(\Delta)$. A natural idea, however, is to partition $M_0$ into subsets $S_1, \ldots, S_{1/\epsilon}$ such that the edges in each subset, more or less, have the same ranks (i.e. a max over min rank ratio of roughly $\Delta^\epsilon$). We can then individually construct a second-stage graph $G_i$ for each $S_i$, find a RGMM $M_i$ of it and use it to augment $S_i$ (and thus $M_0$). Since there are only $1/\epsilon$ groups, there will be one that includes at least $\epsilon = \Omega(1)$ fraction of edges of $M_0$. Therefore, augmenting a constant fraction of edges in this set alone would be enough to break 2 approximation.

However, another technical complication arises here. Once we choose to augment only a subset $S_i$ of the edges in $M_0$, with say $\frac{\max_{e \in S_i}\pi_0(e)}{\min_{e \in S_i}\pi_0(e)} \leq \Delta^\epsilon$, we cannot bound the update-time by  $\widetilde{O}(\Delta^\epsilon)$ anymore. (The argument described before only works if we consider all the edges in $M_0$.) The reason is that, normally, in the second stage graph $G_i$ we would like to have edges that have one endpoint in $S_i$ and one endpoint that is unmatched in $M_0$ so that if both endpoints of an edge $e \in S_i$ are matched in the  matching of $G_i$, then we get a length-3 augmenting path of $M_0$. This makes this second stage graph $G_i$ very sensitive to the precise set of vertices matched/unmatched in the whole matching $M_0$ (as opposed to only those matched in $S_i$) and this would prevent us from using the same argument to bound the update-time by $\widetilde{O}(\Delta^\epsilon)$.

To resolve this issue, on a high level, we also consider any vertex that is matched in $M_0$, but its matching edge has rank higher than those in $S_i$, as ``unmatched'' while constructing graph $G_i$ (see Algorithm~\ref{alg:meta}). This will allow us to argue that the update-time is $\widetilde{O}(\Delta^\epsilon)$. The downside is that not all found length 3 paths will be actual augmenting paths of $M_0$. Fortunately, though, we are still able to argue that the algorithm finds sufficiently many {\em actual} augmenting paths for $M_0$ and thus achieves a $2-\Omega(1)$ approximation (see Section~\ref{sec:approx}).

\section{Preliminaries}\label{sec:preliminaries}

For a graph $G$, we use $\mu(G)$ to denote the size of its maximum matching. Moreover, as standard, for two matchings $M_1$ and $M_2$ of the same graph, we use $M_1 \oplus M_2$ to denote their symmetric difference, i.e., the graph including edges that appear in exactly one of $M_1$ and $M_2$. For two disjoint subsets $A \subseteq V$ and $B \subseteq V$ of the vertex set $V$ of a graph $G$, we use $G[A \times B]$ to denote the bipartite subgraph of $G$ whose vertex set is $A \cup B$ and includes an edge $e$ of $G$ if and only if $e$ has one endpoint in $A$ and one in $B$. Also, generally for a subset $E'$ of the edge-set of $G$, we use $V(E')$ to denote the set of vertices with at least one incident edge in $E'$.

Given a {\em ranking} $\pi$ that maps each edge of $G$ to a real in $[0, 1]$, the greedy maximal matching \lfmm{G, \pi} is obtained as follows: We iterate over the edges in the increasing order of their ranks; upon visiting an edge $e$, if no edge incident to $e$ is already in the matching, $e$ joins the matching. Having this matching, for each edge $e$ we define the eliminator of $e$, denoted by $\elim_{G, \pi}(e)$, to be the edge incident to $e$ that is in matching \lfmm{G, \pi} and has the lowest rank; if $e$ itself is in the matching, then $\elim_{G, \pi}(e) = e$.

The greedy maximal matching algorithm will be particularly useful  if the ranking is random. We do this by picking each entry of the ranking (i.e., the rank of each edge) independently and uniformly at random from $[0, 1]$. Furthermore, it is not hard to see that these ranks need not be too long. Namely, the first $O(\log n)$ bits of the ranks are enough to ensure that no two edges receive the same rank with high probability. From now on, whenever we use the term ``random ranking'' we assume $O(\log n)$ bit ranks from $[0, 1]$ are drawn independently and uniformly at random.

We will make use of the following, by now standard, sparsification property of the random greedy maximal matching algorithm---see e.g. \cite{DBLP:conf/spaa/BlellochFS12,DBLP:conf/icml/AhnCGMW15,DBLP:conf/podc/GhaffariGKMR18,maximalmatchingfocs,assadisoda,misfocs,DBLP:conf/mfcs/Konrad18}.

\begin{lemma}\label{lem:sparsification}
	Fix an arbitrary graph $G$ and a random-ranking $\pi$ on its edge-set. The following event holds w.h.p. over the randomization in $\pi$: For every $O(\log n)$ bit rank $p \in [0, 1]$, the subgraph of $G$ including edges $e$ with $\pi(\elim_{G, \pi}(e)) > p$, has maximum degree $O(p^{-1}\log n)$.
\end{lemma}

\section{A Static Algorithm}

In this section we describe a static algorithm for finding an approximate maximum matching. We show in Section~\ref{sec:approx} that the algorithm provides a $2-\Omega(1)$ approximation and show in Section~\ref{sec:dynamicimplementation} that it can be maintained in update-time $\widetilde{O}(\Delta^\epsilon)$.

\smparagraph{Intuitive explanation of the algorithm.} We start with a RGMM $M_0$. After that, we partition the edge set of $M_0$ into $1/\epsilon$ partitions $S_1, \ldots, S_{1/\epsilon}$ such that roughly the maximum rank over the minimum rank in each partition is at most $\Delta^\epsilon$. Then, focusing on each partition $S_i$, we try augmenting the edges of $S_i$ by finding two random greedy matchings of subgraphs $G'^A_i$ and $G'^B_i$ that are determined based on set $S_i$. Roughly, each edge in $G'^A_i$ (and similarly in $G'^B_i$) has one endpoint that is matched in $S_i$---this is the edge to be augmented---and one endpoint that is either unmatched or matched after the edges in $S_i$ are processed in the greedy construction of $M_0$. Therefore if for an edge $uv \in S_i$, its endpoint $v$ is matched via an edge $vx$ in the matching of $G'^A_i$ and $u$ is matched via an edge $uy$ in the matching of $G'^B_i$, and in addition $x$ and $y$ are unmatched in $M_0$ then $uy, uv, vx$ will be a length 3 augmenting path of $M_0$.



\begin{tboxalg2e}{Meta algorithm for finding a better than 2 approximation of maximum matching.}
\label{alg:meta}
\begin{algorithm}[H]
	\DontPrintSemicolon
	\SetAlgoSkip{bigskip}
	\SetAlgoInsideSkip{}
	
	\Input{Graph $G=(V, E)$ and a parameter $\epsilon \in (0, 1)$.}
	
	$M_0 \gets \lfmm{G, \pi_0}$ where $\pi_0$ is a random permutation of $E$.\label{line:m0}\;
	
	\For{any $i \in \{1, 2, \ldots, 1/\epsilon \}$}{
		\If{$i < 1/\epsilon$}{
			 $S_i \gets \left \{e \mid e \in M_0 \text{ and } \pi_0(e) \in \big(\Delta^{-i\epsilon}, \Delta^{-(i-1)\epsilon}\big]   \right\}$. \tcp*{We may call $S_i$ ``partition $i$''.}
		}\Else{
			$S_{1/\epsilon} \gets \left \{e \mid e \in M_0 \text{ and } \pi_0(e) \in \big[0, \Delta^{-1+\epsilon}\big]   \right\}$.
		}
	$U_i \gets \{u \mid \text{for any edge $uv$ either $uv \not\in M_0$ or $uv \in \cup_{j=1}^{i-1} S_{j}$} \}$.\;
	\tcp{$U_i$ includes nodes that are unmatched in $M_0$ or matched by edges in $S_1, \ldots, S_{i-1}$.}
	
	For each edge $uv \in S_i$ put its endpoint with lower ID in set $V^A_i$ and the other in $V^B_i$.\; 
	\tcp{The use of IDs is just to simplify the statements. Any  arbitrary way of putting one endpoint of the edge in $V^A_i$ and the other in $V^B_i$ would work.}
	
	Partition $U_i$ into $U^A_i$ and $U^B_i$ by each node picking its partition independently and u.a.r.\label{line:partitionU}\;
	
	Sample each edge in $S_i$ independently with probability $p := 0.03$.\label{line:sampleSi}\; 
	
	$V'^A_i \gets$ vertices in $V^A_i$ whose edge in $S_i$ is sampled.\;
	
	$V'^B_i \gets$ vertices in $V^B_i$ whose edge in $S_i$ is sampled.\;
	
	$G'^A_i \gets G[V'^A_i \times U^A_i]$.\;
	$G'^B_i \gets G[V'^B_i \times U^B_i]$.\;
	
	$M_i \gets \lfmm{G'^A_i, \pi_i} \cup \lfmm{G'^B_i, \pi_i}$ where $\pi_i$ is a fresh random permutation of $E$.\label{line:mi}
	}
	
	Return the maximum matching of graph $(M_0 \cup M_1 \cup \ldots \cup M_{ 1/\epsilon})$.
\end{algorithm}
\end{tboxalg2e}

\section{Approximation Factor of Algorithm~\ref{alg:meta}}\label{sec:approx}

In this section, we prove that the approximation factor of Algorithm~\ref{alg:meta} is at most $2-\Omega(1)$ given that $\epsilon$ is a constant.

We fix one arbitrary maximum matching of graph $G$ and denote it by \opt{}; recall that $|\opt| = \mu(G)$. Having this matching \opt{}, we now call an edge $e \in M_0$ {\em 3-augmentable} if it is in a length 3 augmenting path in $\opt \oplus M_0$. Observe that since $\opt$ cannot be augmented, this augmenting path should start and end with edges in $\opt$ and thus edge $e$ has to be in the middle.

The following lemma is crucial in the analysis of the approximation factor. Basically, it says that for any partition $S_i$ where most of edges in $S_i$ are 3-augmentable (which in fact should be the case for most of the partitions if $|M_0|$ is close to $0.5\mu(G)$), roughly $p/4$ fraction of the edges in $S_i$ are in length 3 augmenting paths in $S_i \oplus M_i$. We emphasize that this does not directly prove the bound on the approx factor as these length 3 augmenting paths in $S_i \oplus M_i$ may not necessarily be augmenting paths in $M_0 \oplus M_i$.

\begin{lemma}\label{lem:augmentationinapartition}
	For any $i \in [1/\epsilon]$ and any parameter $\delta \in (0, 1)$, if $(1-\delta)$ fraction of the edges in $S_i$ are 3-augmentable, then in expectation, there are at least $(\frac{(1-\delta)p}{4}-4p^2)|S_i|$ edges in $S_i$ where both of their endpoints are matched in $M_i$.
\end{lemma}

In order to prove this lemma, in Lemma~\ref{lem:greedyvertexsample} we recall a   property of the greedy maximal matching algorithm under vertex samplings originally due to \cite{streamingapprox,streamingarXiv}. We note that the property that we need is slightly stronger than the one proved in \cite[Theorem~3]{streamingarXiv} but follows from a similar argument. Roughly, we need a lower bound on the number of vertices in a specific vertex subset that are matched, while the previous statement only lower bounded the overall matching size. We provide the complete proof in Appendix~\ref{sec:greedyvertexsample}.
%
%

\begin{lemma}\label{lem:greedyvertexsample}
	Let $G(V, U, E)$ be a bipartite graph, $\pi$ be an arbitrary permutation over $E$, and $M$ be an arbitrary matching of $G$. Fix any parameter $p \in (0, 1)$ and let $W$ be a subsample of $V$ including each vertex independently with probability $p$. Define $X$ to be the number of edges in $M$ whose endpoint in $V$ is matched in $\lfmm{G[W \cup U], \pi}$; then $$\E_{W}[X] \geq p(|M|-2p|V|).$$
\end{lemma}

\newcommand{\vm}{\ensuremath{V_M}}
\newcommand{\vmsampled}{\ensuremath{W_{\overline{M}}}}

\newcommand{\vnotm}{\ensuremath{V_{\overline{M}}}}
\newcommand{\vnotmsampled}{\ensuremath{W_{\overline{M}}}}

Equipped with Lemma~\ref{lem:greedyvertexsample}, we are ready to prove Lemma~\ref{lem:augmentationinapartition}.

\begin{proof}[Proof of Lemma~\ref{lem:augmentationinapartition}]
	Fix a partition $S_i$ which includes $(1-\delta)|S_i|$ 3-augmentable edges and denote by $T_i$ the subset of edges in $S_i$ that are 3-augmentable; implying that
	\begin{equation}\label{eq:74881238}
		|T_i| \geq (1-\delta)|S_i|.
	\end{equation}
	Recall that each edge $e \in T_i$ is the middle edge in a length 3 augmenting path in $\opt \oplus M_0$ where $\opt$ is a fixed maximum matching of $G$. Define set $\opt^A$ (resp. $\opt^B$) to be the subset of edges $uv$ in $\opt$ where one of their endpoints, say, $v$ is in $V(T_i) \cap V^A_i$ (resp. $V(T_i) \cap V^B_i$) and the other endpoint $u$ is in set $U^A_i$ (resp. $U^B_i$).
	
	We say an edge $ab \in T_i$ is {\em good} if $a$ is matched in $\opt^A$ and also $b$ is matched in $\opt^B$; and use $Y$ to denote the subset of edges in $T_i$ that are good. We first claim that 
	\begin{equation}\label{eq:124978}
		\E_{U^A_i, U^B_i}[|Y|] \geq \frac{1}{4} |T_i|,
	\end{equation}
	where observe that here the expectation is only taken over the randomization in partitioning $U_i$ into $U^A_i$ and $U^B_i$. To see this, fix an edge $ab \in T_i$ and let $au$ and $bv$ be the edges in $\opt$ that along with $ab$ form a length 3 augmenting path. Observe that edge $au \in \opt^A$ if $u \in U^A_i$ and $bv \in \opt^B$ if $v \in U^B_i$. Since the partition of $v$ and $u$ is chosen independently and u.a.r., there is a probability $\frac{1}{4}$ that both these events occur, implying $ab \in Y$. Linearity of expectation over every edge in $T_i$ proves (\ref{eq:124978}).

	Now, consider graphs $G^A_i := G[V^A_i \times U^A_i]$ and $G^B_i := G[V^B_i \times U^B_i]$ and observe that $\opt^A$ is a matching of $G^A_i$ and $\opt^B$ is a matching of $G^B_i$. One can confirm that $V'^A_i$ is a random subsample of $V^A_i$ where for each $v \in V^A_i$, $\Pr[v \in V'^A_i] = p$. More importantly, whether for a vertex $v$ the event $v \in V'^A_i$ holds is independent of which other vertices are in $V'^A_i$. (Though we note that $v \in V'^A_i$ is not independent of those vertices in $V'^B_i$.) Similarly, $V'^B_i$ can be regarded as a random subsample of $V^B_i$ wherein the vertices appear independently from each other. As a result, graph $G'^A_i$ (resp. $G'^B_i$) is essentially obtained by retaining a random subsample of the vertices in the $V^A_i$ (resp. $V^B_i$) partition of graph $G^A_i$ (resp. $G^B_i$). We can thus use Lemma~\ref{lem:greedyvertexsample} while fixing matching $\opt^A$ to infer that
	\begin{equation}\label{eq:21401234}
	\begin{aligned}
	\E\left[ \begin{tabular}{>{\centering}p{5.5cm}}
	\# of vertices in $V(\opt^A) \cap V'^A_i$ matched in \lfmm{G'^A_i, \pi_i}
	\end{tabular}
	\right] \geq p(|\opt^A|-2p|V^A_i|).
	\end{aligned}
	\end{equation}
	Similarly,
	\begin{equation}\label{eq:281347012}
	\begin{aligned}
	\E\left[ \begin{tabular}{>{\centering}p{5.5cm}}
	\# of vertices in $V(\opt^B) \cap V'^B_i$ matched in \lfmm{G'^B_i, \pi_i}
	\end{tabular}
	\right] \geq p(|\opt^B|-2p|V^B_i|).
	\end{aligned}
	\end{equation}
	Observe that $\E[|V(\opt^A) \cap V'^A_i|] = p|\opt^A|$ since each edge in $\opt^A$ has one endpoint in $V^A_i$ which is sampled to $V'^A_i$ with probability $p$. Combined with (\ref{eq:21401234}) this means that
	\begin{equation}\label{eq:32103470134}
	\begin{aligned}
		\E\left[ \begin{tabular}{>{\centering}p{5.3cm}}
		\# of vertices in $V(\opt^A) \cap V'^A_i$ \underline{not} matched in \lfmm{G'^A_i, \pi_i}
		\end{tabular}
		\right] \leq p|\opt^A| - p(|\opt^A|-2p|V^A_i|) \leq 2p^2|V^A_i| = 2p^2|S_i|.
	\end{aligned}
	\end{equation}
	Similarly by (\ref{eq:281347012}),
	\begin{equation}\label{eq:983014724234}
	\begin{aligned}
	\E\left[ \begin{tabular}{>{\centering}p{5.3cm}}
	\# of vertices in $V(\opt^B) \cap V'^B_i$ \underline{not} matched in \lfmm{G'^B_i, \pi_i}
	\end{tabular}
	\right] \leq p|\opt^B| - p(|\opt^B|-2p|V^B_i|) \leq 2p^2|V^B_i| = 2p^2|S_i|.
	\end{aligned}
	\end{equation}
	By (\ref{eq:124978}) we have $\frac{1}{4}|T_i|$ expected good edges. Out of these, each edge $ab \in Y$ is sampled, i.e., $a \in V'^A_i$ and $b \in V'^B_i$ with probability $p$. Therefore, in expectation, there are a total of $\frac{p}{4}|T_i|$ sampled good edges. Say a sampled good edge $ab$ is wasted if $a$ is unmatched in $\lfmm{G'^A_i, \pi_i}$ or $b$ is unmatched in $\lfmm{G'^B_i, \pi_i}$. Combined with (\ref{eq:32103470134}) and (\ref{eq:983014724234}) there are at most $2p^2|S_i|+2p^2|S_i|\leq4p^2|S_i|$ wasted edges. This means that the expected number of sampled good edges that are not wasted is at least
	$$
	\frac{p}{4}|T_i| - 4p^2|S_i|.
	$$ 
	Moreover, by (\ref{eq:74881238}), $|T_i| \geq (1-\delta)|S_i|$. Replacing this into the equation above, we get that there are, in expectation, at least $\frac{p}{4}(1-\delta)|S_i|-4p^2|S_i| = (\frac{(1-\delta)p}{4}-4p^2)|S_i|$ good edges that are not wasted, i.e., both of their endpoints are matched in $M_i = \lfmm{G'^A_i, \pi_i} \cup \lfmm{G'^B_i, \pi_i}$ as claimed in the lemma.
\end{proof}

The following claim shows that there is a subset $S_{i^\star}$ that is ``large enough'' compared to the size of matching $M_0$ and is much larger than the total number of edges in previous subsets $S_1, \ldots, S_{i^\star - 1}$.

\begin{claim}\label{cl:layerwithlargesize}
	There exists an integer $i^\star \in [1/\epsilon]$ such that 
	$$|S_{i^\star}| \geq \frac{1}{2^{13/\epsilon}} |M_0| \qquad \text{and} \qquad |S_{i^\star}| > 2^{11}\sum_{i=1}^{i^\star-1}|S_i|.$$
\end{claim}
\begin{proof}
	Let $i^\star$ be the smallest integer in $[1/\epsilon]$ for which 
	\begin{equation}\label{eq:18374012384}
	|S_{i^\star}| \geq 2^{12i^\star-\frac{13}{\epsilon}}|M_0|,
	\end{equation}
	we show that both conditions should hold for $i^\star$. First, we have to prove that there is a choice of $i^\star \in [1/\epsilon]$ satisfying (\ref{eq:18374012384}). Suppose for the sake of contradiction that this is not the case; then:
	$$
		\sum_{i=1}^{1/\epsilon} |S_i| < \sum_{i=1}^{1/\epsilon} 2^{12i-\frac{13}{\epsilon}} |M_0| = 2^{\frac{-13}{\epsilon}}|M_0|\sum_{i=1}^{1/\epsilon} 2^{12i} \ll 2^{\frac{-13}{\epsilon}}|M_0|(2 \times 2^{12/\epsilon}) < |M_0|. 
	$$
	Observe that subsets $S_1, \ldots, S_{1/\epsilon}$ partition the edges in $M_0$ and thus it should hold that $\sum_{i=1}^{1/\epsilon} |S_i| = |M_0|$; implying that the equation above is indeed a contradiction, proving existence of $i^\star$.
	
	The first inequality of the claim is automatically satisfied for $i^\star$ due to (\ref{eq:18374012384}) since
	\begin{equation*}
		|S_{i^\star}| \geq 2^{12i^\star-\frac{13}{\epsilon}}|M_0| > 2^{-\frac{13}{\epsilon}}|M_0|.
	\end{equation*}
	It thus only remains to prove the second inequality. For that, observe that since $i^\star$ is the smallest integer satisfying (\ref{eq:18374012384}), then for any $i < i^\star$ we have $|S_i| < 2^{12i-\frac{13}{\epsilon}}|M_0|.$ This means that 
	$$
	\sum_{i=1}^{i^\star-1} |S_i| < \sum_{i=1}^{i^\star-1}2^{12i-\frac{13}{\epsilon}}|M_0| = 2^{\frac{-13}{\epsilon}}|M_0|\sum_{i=1}^{i^\star-1}2^{12i} \ll 2^{\frac{-13}{\epsilon}} |M_0| (2 \times 2^{12(i^\star - 1)}) = 2^{12i^\star-11-\frac{13}{\epsilon}}|M_0|.
	$$
	Combining this with (\ref{eq:18374012384}) we get
	$$
	\frac{|S_{i^\star}|}{\sum_{i=1}^{i^\star - 1}|S_i|} > \frac{2^{12i^\star - \frac{13}{\epsilon}}|M_0|}{2^{12i^\star-11-\frac{13}{\epsilon}}|M_0|} = 2^{11},
	$$
	implying the second inequality of the claim as well.
\end{proof}

Let us recall a folklore property that if maximal matching $M_0$ is not already large enough, then most of the edges in it are 3-augmentable.

\begin{observation}[folklore]\label{obs:length3augpaths}
	If $|M_0| < (\frac{1}{2}+\delta)\mu(G)$, then at least $(\frac{1}{2}-3\delta)\mu(G)$ edges in $M_0$ are 3-augmentable.	
\end{observation}
\begin{proof}
	See e.g. \cite[Lemma 1]{streamingarXiv} for a simple argument.
\end{proof}

We are now ready to analyze the approximation factor. We prove that for $\delta = \frac{1}{1000 \times 2^{13/\epsilon}}$, the matching returned by Algorithm~\ref{alg:meta} has, in expectation, size at least $(\frac{1}{2}+\delta)\mu(G)$. We first assume that $|M_0| < (\frac{1}{2} + \delta)\mu(G)$ as otherwise matching $M_0$ already achieves the desired approximation factor. By Observation~\ref{obs:length3augpaths}, this means that at least $(\frac{1}{2} - 3\delta)\mu(G)$ edges of $M_0$ are 3-augmentable; meaning that the number of edges in $M_0$ that are not 3-augmentable is at most
	\begin{equation}\label{eq:328974}
	|M_0|-\left(\frac{1}{2}-3\delta\right)\mu(G) \leq \left(\frac{1}{2} + \delta\right)\mu(G) - \left(\frac{1}{2}-3\delta\right)\mu(G) = 4\delta\mu(G) \leq 8\delta|M_0|.
	\end{equation}
	Let $i^\star \in [1/\epsilon]$ be the integer satisfying Claim~\ref{cl:layerwithlargesize}. By (\ref{eq:328974}) there are at most $8\delta|M_0|$ edges in $M_0$ and thus in $S_{i^\star}$ that are not 3-augmentable. Therefore, 
	\begin{flalign*}
	\text{\# of 3-augmentable edges in $S_{i^\star}$} &\geq |S_{i^\star}| - 8\delta|M_0| \\
	&= |S_{i^\star}| - 8\frac{1}{1000 \times 2^{13/\epsilon}} |M_0| \\
	&> |S_{i^\star}| - \frac{1}{100} \times \frac{|M_0|}{2^{13/\epsilon}}\\
	&\geq |S_{i^\star}| - \frac{1}{100} \times \frac{2^{13/\epsilon}|S_{i^\star}|}{2^{13/\epsilon}} && \text{First inequality of Claim~\ref{cl:layerwithlargesize}.}\\
	&\geq 0.99 |S_{i^\star}|.
	\end{flalign*}
	Since $0.99$ fraction of the edges in $S_{i^\star}$ are 3-augmentable, by Lemma~\ref{lem:augmentationinapartition}, there are at least 
	$$
	\left(\frac{0.99p}{4}-4p^2\right)|S_{i^\star}| \stackrel{p=0.03}{=} 0.003825 |S_{i^\star}| > 0.003|S_{i^\star}|
	$$ 
	edges in $S_{i^\star}$ whose both endpoints are matched in $M_{i^\star}$. We would like to argue that these form length 3 augmenting paths but note that an edge $e \in M_{i^\star}$ may have an endpoint that is already matched in subsets $S_1, \ldots, S_{i^\star-1}$. However, the crucial observation here is that since by the second inequality of Claim~\ref{cl:layerwithlargesize}, we have $|S_{i^\star}| \geq 2^{11}\sum_{i=1}^{i^\star-1}|S_i|$ and each edge in $S_1 \cup \ldots \cup S_{i^\star-1}$ can be connected to two edges in $M_{i^\star}$ the number of these length 3 augmenting paths that are also augmenting paths in $\opt \oplus M_0$ is at least
	$$
		0.003 |S_{i^\star}| - 2 \times \sum_{i=1}^{i^\star-1}|S_i| \geq 0.003 |S_{i^\star}| - 2 \times \frac{|S_{i^\star}|}{2^{11}} > 0.002 |S_{i^\star}|.
	$$
	Each of these augmenting paths can be used to increase size of $M_0$ by one, therefore the final matching has size at least
	\begin{flalign*}
		|M_0| + 0.002|S_{i^\star}| \geq |M_0| + \frac{0.002}{2^{13/\epsilon}}|M_0| = \left(1+\frac{1}{500\times 2^{13/\epsilon}}\right)|M_0| &\geq \left(1+\frac{1}{500\times 2^{13/\epsilon}}\right)\frac{1}{2}\mu(G)\\
		&= \left(\frac{1}{2}+\frac{1}{1000\times 2^{13/\epsilon}}\right)\mu(G),
	\end{flalign*}
which proves the approximation factor is $2-\Omega(1)$ so long as $\epsilon > 0$ is a constant.
%
%
%
%

\section{Dynamic Implementation of Algorithm~\ref{alg:meta}}\label{sec:dynamicimplementation}

In this section, we describe how we can maintain Algorithm~\ref{alg:meta} in update time $O(\Delta^\epsilon + \polylog n)$.

\subsection{Tools} 
We borrow two black-box tools from the previous works. The first one is a simple corollary of the algorithm of Gupta and Peng \cite{DBLP:conf/focs/GuptaP13}, see also \cite{DBLP:conf/icalp/BernsteinS15} for a proof of this corollary.

\begin{lemma}[{\cite{DBLP:conf/focs/GuptaP13}}]\label{lem:lowdegree}
	Let $\Delta'$ be a fixed upper bound on the maximum degree of a graph at all times. Then we can maintain a $(1+\epsilon)$ approximate matching deterministically under edge insertions and deletions in worst-case update time $O(\Delta'/\epsilon^{2})$ per update.
\end{lemma}

We use Lemma~\ref{lem:lowdegree} only for the last step of Algorithm~\ref{alg:meta} in which we need to maintain a maximum matching of $M_0 \cup M_1 \cup \ldots \cup M_{1/\epsilon}$ which is a graph with maximum degree $O(1/\epsilon)$.

The second black-box result that we use is due to the recent algorithm of Behnezhad~\etal{}~\cite{misfocs} that asserts a random greedy maximal matching can be maintained efficiently.

\begin{lemma}[{\cite{misfocs}}]\label{lem:greedyMM}
	Let $\pi$ be a random ranking where $\pi(e) \in [0, 1]$ for each edge $e$ is drawn uniformly at random upon its arrival. Then maximal matching \lfmm{G, \pi} can be maintained under edge insertions and deletions in expected time $O(\log^2 \Delta \times \log^2 n)$ per update (without amortization). Furthermore, for each update, the adjustment-complexity is in expectation $O(1)$ and w.h.p. $O(\log n)$.
\end{lemma}

%
%

\subsection{Data Structures \& Setup}

\newcommand{\issampled}[0]{\ensuremath{\text{{\normalfont\textsf{ is\_sampled}}}}}
\newcommand{\partition}[0]{\ensuremath{\text{{\normalfont \textsf{partition}}}}}

Algorithm~\ref{alg:meta} computes two types of matchings: (1) Matching $M_0 = \lfmm{G, \pi_0}$ which is a standard random greedy maximal matching of the whole graph $G$. (2) Matchings $M_1, \ldots, M_{1/\epsilon}$ which are computed on specific subgraphs of $G$. Observe in Algorithm~\ref{alg:meta} that each matching $M_i$ for $i \in \{1, \ldots, 1/\epsilon\}$ is the union of two random greedy matchings $\lfmm{G'^A_i, \pi_i}$ and $\lfmm{G'^B_i, \pi_i}$. A crucial observation here is that these two graphs $G'^A_i$ and $G'^B_i$ by definition are vertex disjoint. Therefore defining graph $G_i$ to be the union of these two graphs, $M_i$ would be equivalent to \lfmm{G_i, \pi_i}.

Now we have $\frac{1}{\epsilon}+1$ graphs $G, G_1, \ldots, G_{1/\epsilon}$ and $\frac{1}{\epsilon}+1$ independently drawn rankings $\pi_0, \pi_1, \ldots, \pi_{1/\epsilon}$. Therefore, if a priori these graphs were fixed and remained unchanged after each edge insertion/deletion, we could use Lemma~\ref{lem:greedyMM} to update each one of them in expected time $\polylog n$ requiring only a total update-time of $O(\frac{1}{\epsilon}) \cdot \polylog n$. However, as highlighted in Section~\ref{sec:overview} the challenge is that the vertex sets of graphs $G_1, \ldots, G_{1/\epsilon}$ are {\em adaptively} determined based on matching $M_0$. That is, a single edge update that changes matching $M_0$ may lead to many {\em vertex} insertions/deletions to graphs $G_1, \ldots, G_{1/\epsilon}$ that are generally much harder to handle than edge updates. Therefore, we need to be careful about what to maintain and how to do it to ensure these vertex updates can be determined and handled efficiently.

\smparagraph{Fixing the randomizations.} To maintain the matching of Algorithm~\ref{alg:meta}, we fix all the randomizations required. There are two types of randomizations involved: (1) Randomizations on the edges, such as the random rankings and edge samplings; as in Lines \ref{line:m0},\ref{line:mi}, and \ref{line:sampleSi}. (2) Randomizations on the vertices; as in Line~\ref{line:partitionU}. We reveal the randomizations on the vertex set in the preprocessing step as it is static. But we reveal the randomizations on the edges upon their arrival. For completeness, we mention the precise random bits drawn below. 

For each edge $e$, we draw the following upon its arrival:

\begin{highlighttechnical}
	\begin{description}
		\item $\issampled_i(e) \in \{0, 1\}$: This is drawn for any $i \in [1/\epsilon]$ independently. It is  $1$ with probability $p$ and $0$ otherwise. It determines the outcome of edge-sampling in Line~\ref{line:sampleSi} of the algorithm.
		\item $\pi_i(e) \in [0, 1]$: The rank of $e$ in ranking $\pi_i$. This is drawn for any $i \in \{0, \ldots, 1/\epsilon\}$.
	\end{description} 
\end{highlighttechnical}

And for each vertex $v$, in the pre-processing step, we draw:

\begin{highlighttechnical}
	\begin{description}
		\item $\partition_i(v) \in \{A, B\}$: This is drawn for any $i \in [1/\epsilon]$ independently. It is $A$ with probability $0.5$ and $B$ otherwise. The value determines whether $v$ would join $U^A_i$ or $U^B_i$ if it is partitioned in Line~\ref{line:partitionU} of the algorithm.
	\end{description} 
\end{highlighttechnical}

\smparagraph{Data structures.} Let us for simplicity define $G_0 = G$. For any vertex $v$ and any $i \in \{0, \ldots, \frac{1}{\epsilon}\}$ we maintain the following data structures:

\begin{highlighttechnical}
	\begin{description}
		\item $k_i(v)$: If $v$ is not part of graph $G_i$ or if it is unmatched in \lfmm{G_i, \pi_i} then $k_i(v) = 1$. Otherwise, if $e$ is the edge incident to $v$ that is in matching $\lfmm{G_i, \pi_i}$ then $k_i(v) = \pi_i(e).$
		\item $N_i(v)$: The set of neighbors of $v$ in graph $G_i$. This set is stored as a self-balancing binary search tree in which each neighbor $u$ of $v$ is indexed by $\pi_i(\elim_{G_i, \pi_i}(uv))$. (If $v$ is not in the vertex-set of $G_i$ then simply $N_i(v) = \emptyset$.)
	\end{description} 
\end{highlighttechnical}

\subsection{The Update Algorithm} 

We run $1+1/\epsilon$ instances of Lemma~\ref{lem:greedyMM} for maintaining greedy matchings of $G_0, \ldots, G_{1/\epsilon}$. Moreover, we run a single instance of Lemma~\ref{lem:lowdegree} on the edges in union of matchings $M_0 \cup \ldots \cup M_{1/\epsilon}$. 

As mentioned previously, a single edge update to graph $G_0$ may change the structure of graphs $G_1, \ldots, G_{1/\epsilon}$ and in particular may lead to vertex insertions or deletions in them. Therefore, our main focus in this section is to show how we can detect these vertices that join/leave graphs $G_1, \ldots, G_{1/\epsilon}$ and their incident edges in these graphs efficiently. Before that, we need the following lemma. The proof is a simple consequence of Lemma~\ref{lem:sparsification} and thus we defer it to Appendix~\ref{sec:missingproofs}.

\begin{lemma}\label{lem:updatedsafteredgeupdate}
	Suppose that an edge $e$ is inserted to or deleted from a graph $G_i$ for some $i \in \{0, \ldots, 1/\epsilon\}$. After updating matching $\lfmm{G_i, \pi_i}$ (e.g. by Lemma~\ref{lem:greedyMM}) and getting the list $L$ of edges that joined or left the matching, we can update $k_i(\cdot)$ and $N_i(\cdot)$ accordingly in expected time $\polylog n$.
\end{lemma}

Consider insertion or deletion of an edge $f$. We use the following procedure to maintain our data structures and finally the matching returned by Algorithm~\ref{alg:meta}.

\smparagraph{Step 1: Updating $M_0$.} We first update matching $M_0$. This is done by  Lemma~\ref{lem:greedyMM} in $\polylog n$ expected time. After that, we also update data structures $k_0(v)$ and $N_0(v)$ where necessary using Lemma~\ref{lem:updatedsafteredgeupdate}. There are two cases. If matching $M_0$ changes after the update, then we may have to update the vertex sets of graphs $G_1, \ldots, G_{1/\epsilon}$. This is the operation that is costly and we handle it in the next steps. If $M_0$ does not change, the only remaining update is to see if $f$ itself is part of a graph $G_i$ and reflect that. This only takes polylogarithmic time using Lemmas~\ref{lem:greedyMM} and \ref{lem:updatedsafteredgeupdate}.

\smparagraph{Step 2:  Updating vertex-sets of $G_1, \ldots, G_{1/\epsilon}$.} The vertex-set of each graph $G_i$ is composed of four disjoint subsets $V'^A_i, U^A_i, V'^B_i, $ and $U^B_i$. One can confirm from Algorithm~\ref{alg:meta} that whether a vertex $v$ belongs to one of these sets (and which one if so) can be uniquely determined by knowing the edge incident to $v$ that is in matching $M_0$ or knowing that no such edge exists. Therefore:
\begin{observation}\label{obs:edgetovertexupdate}
	If after the update, a vertex $v$ leaves or is added to the vertex-set of a graph $G_i$, then there must exist an edge connected to $v$ that either joined or left matching $M_0$.
\end{observation}

By Observation~\ref{obs:edgetovertexupdate}, to update the vertex-sets, it suffices to only  iterate over vertices whose matching edge in $M_0$ has changed and determine which graph $G_i$ they should belong to. The procedure is a simple consequence of the way Algorithm~\ref{alg:meta} constructs these graphs and also the randomizations fixed previously. We provide the details in Algorithm~\ref{alg:updatevertexsets} for completeness.

It has to be noted that we are only updating the vertex-sets in this step. In particular, for a vertex $v$ that e.g. joins graph $G_i$, we do not construct its adjacency list $N_i(v)$ yet. This is postponed to the next step after all the vertex sets are completely updated.

\begin{tboxalg2e}{Updating vertex-sets of $G_1, \ldots, G_{1/\epsilon}$.}
\label{alg:updatevertexsets}
\begin{algorithm}[H]
	\DontPrintSemicolon
	\SetAlgoSkip{bigskip}
	\SetAlgoInsideSkip{}
	
	\For{any vertex $v$ whose match-status in $M_0$ has changed after the update}{
		\If{$v$ is now unmatched}{
			$\ell_v \gets 0$.
		}
		\Else{
			Let $e$ be the edge incident to $v$ that is now in matching $M_0$.\;
			Let $S_j$ be the partition to which $e$ will be assigned in Algorithm~\ref{alg:meta} based on $\pi_0(e)$.\;
			$\ell_v \gets j$.
		}
		\For{any $i \in [1/\epsilon]$}{
			\If{$\ell_v < i$}{
				If $\partition_i(v) = A$, then $v \in U^A_i$. Otherwise $\partition_i(v) = B$, thus $v \in U^B_i$.\;
			}
			\If{$\ell_v = i$}{
				\tcp{At this state, $v$ should be matched in $M_0$ through its incident edge $e$.}
				\If{$\issampled_i(e) = 0$}{
					Vertex $v$ is not in the vertex-set of graph $G_i$.
				}\Else{
					If $v$ is the lower-ID endpoint of $e$, then $v \in V'^A_i$. Otherwise, $v \in V'^B_i$.
				}
			}
			\If{$\ell_v > i$}{
				Vertex $v$ is not in the vertex-set of graph $G_i$.
			}
		}
	}
\end{algorithm}
\end{tboxalg2e}

\smparagraph{Step 3: Updating adjacency lists of $G_1, \ldots, G_{1/\epsilon}$ and their matchings.} The previous step updated the vertex-sets. Here, we update the adjacency lists and the matchings $M_1, \ldots, M_{1/\epsilon}$. Precisely, we update data structure $N_i(v)$ for each vertex $v$ and each $i \in [1/\epsilon]$ where necessary. Note that for any vertex $v$, both $k_0(v)$ and adjacency list $N_0(v)$ were already updated in Step 1.

First, for any vertex $v$ that leaves a graph $G_i$, we immediately remove its incident edges from the graph one by one. Each one of these should be regarded as edge deletions and thus we can use Lemma~\ref{lem:greedyMM} to update $M_i$. We then update $k_i$ and $N_i$ data structures accordingly using Lemma~\ref{lem:updatedsafteredgeupdate}.

Next, for any vertex $v$ that is added to the vertex set of a graph $G_i$, we have to determine the set of its neighbors in this graph. To do so, we take the steps  formalized as Algorithm~\ref{alg:updateadjacencylists}. A crucial observation to note before reading the description of Algorithm~\ref{alg:updateadjacencylists} is stated below. The proof is a direct consequence of the greedy structure of RGMM, thus we defer it to Appendix~\ref{sec:missingproofs}.

\begin{claim}\label{cl:gkremainsunchanged}
Suppose that the edge $f$ that is being inserted to/deleted from $G$ is part of matching $M_0$ (if deleted before deletion and if inserted after insertion). Note that if this was not the case, then updating $f$ would not change the vertex sets of $G_1, \ldots, G_k$. Also assume that $f$ belongs to partition $S_j$ of matching $M_0$ in Algorithm~\ref{alg:meta}. Then this update may only affect vertex sets of graphs $G_1, \ldots, G_j$. In particular, any graph $G_k$ with $k > j$ remains unchanged after insertion or deletion of $f$.
\end{claim}

Claim~\ref{cl:gkremainsunchanged} is algorithmically useful in the following way. Suppose that $f \in S_j$ and let $\alpha$ be the minimum rank considered to be in $S_j$ in Algorithm~\ref{alg:meta}. Then we can remove all edges whose eliminator ranks are less than $\alpha$ from $G$, and the remaining graph will include all edges that we have to consider for graphs $G_1, \ldots, G_j$. This, by Lemma~\ref{lem:sparsification} prunes the degrees to $\widetilde{O}(\alpha^{-1})$ and helps reducing the running time. See Algorithm~\ref{alg:updateadjacencylists} and Lemma~\ref{lem:correctnessandrunningtime} for the details.

\begin{tboxalg2e}{Updating adjacency lists of $G_1, \ldots, G_{1/\epsilon}$.}
\label{alg:updateadjacencylists}
\begin{algorithm}[H]
	\DontPrintSemicolon
	\SetAlgoSkip{bigskip}
	\SetAlgoInsideSkip{}
	
	Let $f$ be the original edge that was inserted/deleted from $G$ and suppose that it changed matching $M_0$ (otherwise, vertex-sets of $G_1, \ldots, G_{1/\epsilon}$ will remain the same.)\;
	
	Suppose that $f \in S_j$ (if $f$ was deleted, $f \in S_j$ before deletion, and if inserted, $f \in S_j$ after it).\;
	\tcp{Note that $f$ has to be in $M_0$ to change it once updated. Thus it should belong to a set $S_j$.}
	
	If $j < 1/\epsilon$ then let $\alpha \gets \Delta^{-i\epsilon}$, otherwise if $j = 1/\epsilon$ let $\alpha \gets 0$.\;
	\tcp{$\alpha$ is the lower bound on edge ranks that get partitioned to $S_j$ according to Algorithm~\ref{alg:meta}.}
	
	\For{any vertex $v$ and any $i \in [1/\epsilon]$ such that $v$ joins the vertex set of $G_i$}{
		\tcp{We can detect these vertices efficiently by only going through the changes found in Step 2 without exhaustively checking all vertices in the graph.}
		
		$L_v \gets \{u \in N_0(v) \mid \pi_0(\elim_{G_0, \pi_0}(uv)) \geq \alpha \}$\;
		
		\tcp{Set $L_v$ has size $\min\{\Delta, O(\alpha^{-1}\log n)\}$ by Lemma~\ref{lem:sparsification} and can be constructed in time $\widetilde{O}(|L_v|)$ since all edges in $N_0(v)$ are already indexed by their eliminator ranks.}
		
		\For{any neighbor $u \in L_v$ of $v$}{
			\If{$(v \in V'^A_i \text{ and } u \in U^A_i)$ or $(v \in V'^B_i \text{ and } u \in U^B_i)$ or $(u \in V'^A_i \text{ and } v \in U^A_i)$ or $(u \in V'^B_i \text{ and } v \in U^B_i)$}{
				Add $u$ to $N_i(v)$ and $v$ to $N_i(u)$.\;
				Update matching $M_i$ using Lemma~\ref{lem:greedyMM} according to this edge insertion.\;
				Update $k_i(\cdot)$ and $N_i(\cdot)$ as necessary by this edge insertion using Lemma~\ref{lem:updatedsafteredgeupdate}.
			}
		}
	}
\end{algorithm}
\end{tboxalg2e}

\smparagraph{Step 4: Updating the final matching.} Finally, recall that we run multiple instances of Lemma~\ref{lem:lowdegree} to maintain a $(1+\epsilon)$ approximate maximum matching of graph $M_0 \cup \ldots \cup M_{1/\epsilon}$ which will include our final matching. Throughout the updates above, we keep track of all edges that leave/join these matchings and for each one of them we update this final matching via Lemma~\ref{lem:lowdegree}.

\subsection{Correctness \& Running Time of Update Algorithm}

In this section, as the title describes, we prove the correctness of the update algorithm above and analyze its running time. Namely, we prove the following lemma.

\begin{lemma}\label{lem:correctnessandrunningtime}
	The update algorithm of previous section correctly updates all data structures and the matching and its expected running time per update without amortization is $O(\Delta^\epsilon \polylog n)$.
\end{lemma}

Before that, let us show how we can actually turn this update-time to $O(\Delta^\epsilon + \polylog n)$ as claimed by Theorem~\ref{thm:main}. To do so, given $\epsilon$, we consider a smaller value for $\epsilon$, say $\epsilon/2$. Then the update-time would be $O(\Delta^{\epsilon/2}\polylog n)$. Now if $\Delta^{\epsilon/2} \gg \polylog n$, then we already have $\Delta^{\epsilon/2}\polylog n \ll \Delta^\epsilon$. Otherwise, $\Delta$ is polylogarithmic and the whole update-time is also polylogarithmic.

As another note, in Theorem~\ref{thm:main} we state that the update-time is worst-case but Lemma~\ref{lem:correctnessandrunningtime} bounds the expected update-time. To turn this into a worst-case bound, we use the reduction of Bernstein~\etal{} \cite{DBLP:conf/soda/BernsteinFH19}. For the reduction to work, the crucial property is that the update-time bound should hold in expectation but without any amortization, as is the case here.

\begin{proof}[Proof of Lemma~\ref{lem:correctnessandrunningtime}]
It is easy to verify correctness of Steps 1, 2, and 4 which are actually quite fast and take only $\polylog n$ time in total. We do provide the necessary details for these steps at the end of this proof. However, the main component of the update-algorithm is Step~3 which takes $O(\Delta^\epsilon \polylog n)$ time. We thus first focus on this step and analyze its running time and correctness.

As before, assume that edge $f$ is updated. If matching $M_0$ does not change as a result of this update, then the vertex sets of all graphs $G_1, \ldots, G_{1/\epsilon}$ will remain unchanged. However, if updating $f$ changes $M_0$, then $f$ should be in $M_0$ once in the graph. As in Algorithm~\ref{alg:updateadjacencylists}, we assume $f \in S_j$ and let $\alpha$ be the minimum possible rank in $S_j$. By Claim~\ref{cl:gkremainsunchanged}, graphs $G_{j+1}, \ldots, G_{1/\epsilon}$ remain unchanged. Also, one can confirm from Algorithm~\ref{alg:meta} that all edges in graphs $G_1, \ldots, G_j$ have eliminator rank of at least $\alpha$ in $M_0$. Thus, for any vertex $v$ added to a graph $G_i$, the set $L_v$ indeed includes all edges incident to $v$ that may belong to $G_i$. Moreover, by Lemma~\ref{lem:sparsification} this set $L_v$ has size at most $\min\{\Delta, O(\alpha^{-1}\log n)\}$ and that can be found in time $\widetilde{O}(|L_v|)$ since the neighbors of $v$ in $N_0(v)$ are indexed by their eliminator-rank. The overall update-time required for Step~3 is thus
\begin{flalign*}
	\text{(\# of edges updated in $M_0$)} \times \min\left\{\Delta, O\left(\frac{\log n}{\alpha}\right)\right\}.
\end{flalign*}
By Lemma~\ref{lem:greedyMM}, the number of edges that are updated in $M_0$ is w.h.p. bounded by $O(\log n)$. It remains to determine the expected value of the second factor in the running time above. Let us use $I_1, \ldots, I_{1/\epsilon}$ to denote the interval of ranks considered by Algorithm~\ref{alg:meta}. That is, $I_{1/\epsilon} = [0, \Delta^{-1+\epsilon}]$ and for any $i < 1/\epsilon$, $I_i = (\Delta^{-i\epsilon}, \Delta^{-(i-1)\epsilon}]$. For edge $f$ that is to be updated, probability that $\pi_0(f)$ is in the $i$th interval is upper bounded by $\Delta^{-(i-1)\epsilon}$. Moreover, given that $\pi_0(f)$ is in the $i$th interval, then $\min\{\Delta, \alpha^{-1}\log n\}$ would be at most $\Delta^{i\epsilon}\log n$. Thus:
\begin{flalign*}
\E\left[\min\left\{\Delta, O\left(\frac{\log n}{\alpha}\right)\right\}\right] &\leq \sum_{i = 1}^{1/\epsilon} \Pr[\pi_0(f) \in I_i] \times \E\left[\min\left\{\Delta, O\left(\frac{\log n}{\alpha}\right)\right\} \,\Bigg\vert\, \pi_0(f) \in I_i \right]\\
	& \leq \sum_{i = 1}^{1/\epsilon} \Delta^{-(i-1)\epsilon} \times \Delta^{i\epsilon} \log n \leq \Delta^\epsilon \log n.
\end{flalign*}
This means that the overall update-time required for Step~3 is $O(\Delta^\epsilon \polylog n)$.

Now, we focus on the other steps.

In Step~1, only matching $M_0$ as well as the data structures related to it are updated. These are correct and only take $\polylog n$ expected time by Lemmas~\ref{lem:greedyMM} and \ref{lem:updatedsafteredgeupdate}.

In Step~2, we detect the updates to the vertex sets of $G_1, \ldots, G_{1/\epsilon}$. This is done in Algorithm~\ref{alg:updatevertexsets} by iterating over all vertices whose match-status in $M_0$ is changed and checking the conditions of Algorithm~\ref{alg:meta}. By Observation~\ref{obs:edgetovertexupdate}, indeed any vertex who gets added/deleted from any graph $G_i$ should have an incident edge in $M_0$ whose match-status in $M_0$ has changed. Therefore, we do discover all the updates. Moreover, for each vertex encountered we only spend $O(1)$ time in Algorithm~\ref{alg:updatevertexsets} to detect which graph $G_i$ it belongs to and there are w.h.p. only $O(\log n)$ such vertices who have an edge with an updated match-status in $M_0$ by Lemma~\ref{lem:greedyMM}. Thus, the total running time for Step~2 is $O(\log n)$.

Finally, in Step~4, we update the final matching. The graph here is composed of $O(1/\epsilon)$ matchings and thus has maximum degree $O(1/\epsilon)$. Therefore, each update by Lemma~\ref{lem:lowdegree} takes $O(1/\epsilon) = O(1)$ time. Moreover, an edge update to graph $M_0$ w.h.p. affects at most $O(\log n)$ edges by Lemma~\ref{lem:greedyMM}, each of these updated edges may lead to vertex insertions/deletions in each graph $G_i$. But these propagated vertex insertions/deletions also affect at most $O(\log n)$ edges in each of the graphs. Thus, the total number of edges that leave/join graph $M_1 \cup \ldots \cup M_{1/\epsilon}$ is at most $O(\log^2 n)$, which is also the upper bound on the running time of Step~4.
\end{proof}

\section*{Acknowledgements}

Soheil Behnezhad thanks Mahsa Derakhshan, MohammadTaghi Hajiaghayi, Cliff Stein and Madhu Sudan for their collaboration in \cite{misfocs} and Sepehr Assadi for many helpful discussions and pointers to the literature.

\bibliographystyle{plain}
\bibliography{refs}

\begin{thebibliography}{10}

\bibitem{DBLP:conf/icml/AhnCGMW15}
Kook~Jin Ahn, Graham Cormode, Sudipto Guha, Andrew McGregor, and Anthony Wirth.
\newblock {Correlation Clustering in Data Streams}.
\newblock In {\em Proceedings of the 32nd International Conference on Machine
  Learning, {ICML}}, pages 2237--2246, 2015.

\bibitem{DBLP:conf/icalp/ArarCCSW18}
Moab Arar, Shiri Chechik, Sarel Cohen, Cliff Stein, and David Wajc.
\newblock {Dynamic Matching: Reducing Integral Algorithms to
  Approximately-Maximal Fractional Algorithms}.
\newblock In {\em 45th International Colloquium on Automata, Languages, and
  Programming, {ICALP} 2018, July 9-13, 2018, Prague, Czech Republic}, pages
  7:1--7:16, 2018.

\bibitem{assadisoda}
Sepehr Assadi, Krzysztof Onak, Baruch Schieber, and Shay Solomon.
\newblock {Fully Dynamic Maximal Independent Set with Sublinear in $n$ Update
  Time}.
\newblock In {\em Proceedings of the Thirtieth Annual {ACM-SIAM} Symposium on
  Discrete Algorithms, {SODA} 2019, San Diego, California, USA, January 6-9,
  2019}, pages 1919--1936, 2019.

\bibitem{DBLP:conf/focs/BaswanaGS11}
Surender Baswana, Manoj Gupta, and Sandeep Sen.
\newblock {Fully Dynamic Maximal Matching in $O(\log n)$ Update Time}.
\newblock In {\em {IEEE} 52nd Annual Symposium on Foundations of Computer
  Science, {FOCS} 2011, Palm Springs, CA, USA, October 22-25, 2011}, pages
  383--392, 2011.

\bibitem{DBLP:journals/siamcomp/BaswanaGS18}
Surender Baswana, Manoj Gupta, and Sandeep Sen.
\newblock {Fully Dynamic Maximal Matching in $O(\log n)$ Update Time (Corrected
  Version)}.
\newblock {\em {SIAM} J. Comput.}, 47(3):617--650, 2018.

\bibitem{misfocs}
Soheil Behnezhad, Mahsa Derakhshan, MohammadTaghi Hajiaghayi, Cliff Stein, and
  Madhu Sudan.
\newblock {Fully Dynamic Maximal Independent Set in Polylogarithmic Update
  Time}.
\newblock In {\em 60th Annual {IEEE} Symposium on Foundations of Computer
  Science, {FOCS} 2019, to appear.}

\bibitem{maximalmatchingfocs}
Soheil Behnezhad, MohammadTaghi Hajiaghayi, and David~G. Harris.
\newblock {Exponentially Faster Massively Parallel Maximal Matching}.
\newblock In {\em 60th Annual {IEEE} Symposium on Foundations of Computer
  Science, {FOCS} 2019, to appear.}

\bibitem{DBLP:conf/soda/BernsteinFH19}
Aaron Bernstein, Sebastian Forster, and Monika Henzinger.
\newblock {A Deamortization Approach for Dynamic Spanner and Dynamic Maximal
  Matching}.
\newblock In {\em Proceedings of the Thirtieth Annual {ACM-SIAM} Symposium on
  Discrete Algorithms, {SODA} 2019, San Diego, California, USA, January 6-9,
  2019}, pages 1899--1918, 2019.

\bibitem{DBLP:conf/icalp/BernsteinS15}
Aaron Bernstein and Cliff Stein.
\newblock {Fully Dynamic Matching in Bipartite Graphs}.
\newblock In {\em Automata, Languages, and Programming - 42nd International
  Colloquium, {ICALP} 2015, Kyoto, Japan, July 6-10, 2015, Proceedings, Part
  {I}}, pages 167--179, 2015.

\bibitem{DBLP:conf/soda/BernsteinS16}
Aaron Bernstein and Cliff Stein.
\newblock {Faster Fully Dynamic Matchings with Small Approximation Ratios}.
\newblock In {\em Proceedings of the Twenty-Seventh Annual {ACM-SIAM} Symposium
  on Discrete Algorithms, {SODA} 2016, Arlington, VA, USA, January 10-12,
  2016}, pages 692--711, 2016.

\bibitem{DBLP:conf/ipco/BhattacharyaCH17}
Sayan Bhattacharya, Deeparnab Chakrabarty, and Monika Henzinger.
\newblock {Deterministic Fully Dynamic Approximate Vertex Cover and Fractional
  Matching in {$O(1)$} Amortized Update Time}.
\newblock In {\em Integer Programming and Combinatorial Optimization - 19th
  International Conference, {IPCO} 2017, Waterloo, ON, Canada, June 26-28,
  2017, Proceedings}, pages 86--98, 2017.

\bibitem{DBLP:journals/siamcomp/BhattacharyaHI18}
Sayan Bhattacharya, Monika Henzinger, and Giuseppe~F. Italiano.
\newblock {Deterministic Fully Dynamic Data Structures for Vertex Cover and
  Matching}.
\newblock {\em {SIAM} J. Comput.}, 47(3):859--887, 2018.

\bibitem{DBLP:conf/stoc/BhattacharyaHN16}
Sayan Bhattacharya, Monika Henzinger, and Danupon Nanongkai.
\newblock {New Deterministic Approximation Algorithms For Fully Dynamic
  Matching}.
\newblock In {\em Proceedings of the 48th Annual {ACM} {SIGACT} Symposium on
  Theory of Computing, {STOC} 2016, Cambridge, MA, USA, June 18-21, 2016},
  pages 398--411, 2016.

\bibitem{BhattacharyaArxiv}
Sayan Bhattacharya, Monika Henzinger, and Danupon Nanongkai.
\newblock {New Deterministic Approximation Algorithms for Fully Dynamic
  Matching}.
\newblock {\em CoRR}, abs/1604.05765, 2016.

\bibitem{DBLP:conf/soda/BhattacharyaHN17}
Sayan Bhattacharya, Monika Henzinger, and Danupon Nanongkai.
\newblock {Fully Dynamic Approximate Maximum Matching and Minimum Vertex Cover
  $O(\log^3 n)$ Worst Case Update Time}.
\newblock In {\em Proceedings of the Twenty-Eighth Annual {ACM-SIAM} Symposium
  on Discrete Algorithms, {SODA} 2017, Barcelona, Spain, Hotel Porta Fira,
  January 16-19}, pages 470--489, 2017.

\bibitem{DBLP:conf/spaa/BlellochFS12}
Guy~E. Blelloch, Jeremy~T. Fineman, and Julian Shun.
\newblock {Greedy Sequential Maximal Independent Set and Matching are Parallel
  on Average}.
\newblock In {\em 24th {ACM} Symposium on Parallelism in Algorithms and
  Architectures, {SPAA} '12, Pittsburgh, PA, USA, June 25-27, 2012}, pages
  308--317, 2012.

\bibitem{DBLP:conf/icalp/CharikarS18}
Moses Charikar and Shay Solomon.
\newblock {Fully Dynamic Almost-Maximal Matching: Breaking the Polynomial
  Worst-Case Time Barrier}.
\newblock In {\em 45th International Colloquium on Automata, Languages, and
  Programming, {ICALP} 2018, July 9-13, 2018, Prague, Czech Republic}, pages
  33:1--33:14, 2018.

\bibitem{DBLP:conf/podc/GhaffariGKMR18}
Mohsen Ghaffari, Themis Gouleakis, Christian Konrad, Slobodan Mitrovic, and
  Ronitt Rubinfeld.
\newblock {Improved Massively Parallel Computation Algorithms for MIS,
  Matching, and Vertex Cover}.
\newblock In {\em Proceedings of the 2018 {ACM} Symposium on Principles of
  Distributed Computing, {PODC} 2018, Egham, United Kingdom, July 23-27, 2018},
  pages 129--138, 2018.

\bibitem{DBLP:conf/focs/GuptaP13}
Manoj Gupta and Richard Peng.
\newblock {Fully Dynamic {$(1+\epsilon)$}-Approximate Matchings}.
\newblock In {\em 54th Annual {IEEE} Symposium on Foundations of Computer
  Science, {FOCS} 2013, 26-29 October, 2013, Berkeley, CA, {USA}}, pages
  548--557, 2013.

\bibitem{DBLP:conf/mfcs/Konrad18}
Christian Konrad.
\newblock A simple augmentation method for matchings with applications to
  streaming algorithms.
\newblock In {\em 43rd International Symposium on Mathematical Foundations of
  Computer Science, {MFCS} 2018, August 27-31, 2018, Liverpool, {UK}}, pages
  74:1--74:16, 2018.

\bibitem{streamingarXiv}
Christian Konrad, Fr{\'{e}}d{\'{e}}ric Magniez, and Claire Mathieu.
\newblock Maximum matching in semi-streaming with few passes.
\newblock {\em CoRR}, abs/1112.0184, 2011.

\bibitem{streamingapprox}
Christian Konrad, Fr{\'{e}}d{\'{e}}ric Magniez, and Claire Mathieu.
\newblock Maximum matching in semi-streaming with few passes.
\newblock In {\em Approximation, Randomization, and Combinatorial Optimization.
  Algorithms and Techniques - 15th International Workshop, {APPROX} 2012, and
  16th International Workshop, {RANDOM} 2012, Cambridge, MA, USA, August 15-17,
  2012. Proceedings}, pages 231--242, 2012.

\bibitem{DBLP:books/daglib/0012859}
Michael Mitzenmacher and Eli Upfal.
\newblock {\em Probability and computing - randomized algorithms and
  probabilistic analysis}.
\newblock Cambridge University Press, 2005.

\bibitem{DBLP:conf/stoc/NeimanS13}
Ofer Neiman and Shay Solomon.
\newblock {Simple deterministic algorithms for fully dynamic maximal matching}.
\newblock In {\em Symposium on Theory of Computing Conference, STOC'13, Palo
  Alto, CA, USA, June 1-4, 2013}, pages 745--754, 2013.

\bibitem{DBLP:conf/stoc/OnakR10}
Krzysztof Onak and Ronitt Rubinfeld.
\newblock {Maintaining a large matching and a small vertex cover}.
\newblock In {\em Proceedings of the 42nd {ACM} Symposium on Theory of
  Computing, {STOC} 2010, Cambridge, Massachusetts, USA, 5-8 June 2010}, pages
  457--464, 2010.

\bibitem{DBLP:conf/focs/Solomon16}
Shay Solomon.
\newblock {Fully Dynamic Maximal Matching in Constant Update Time}.
\newblock In {\em {IEEE} 57th Annual Symposium on Foundations of Computer
  Science, {FOCS} 2016, 9-11 October 2016, Hyatt Regency, New Brunswick, New
  Jersey, {USA}}, pages 325--334, 2016.

\end{thebibliography}

\appendix

\section{Greedy Matching Size under Vertex Sampling}\label{sec:greedyvertexsample}

In this section, we prove Lemma~\ref{lem:greedyvertexsample} by extending the ideas presented in \cite{streamingarXiv}.

\begin{proof}
	Consider the following equivalent process of constructing $\lfmm{G[W \cup U], \pi}$ that gradually reveals the subsample $W$ of $V$: We initialize matching $M \gets \emptyset$, initially mark each vertex as {\em alive}, and then iterate over the edges in $E$ in the order of $\pi$. Upon visiting an edge $vu$ with $v \in V$ and $u \in U$, if either $v$ or $u$ is {\em dead} (i.e., not alive) we discard $vu$. Otherwise, we call $vu$ a {\em potential-match} and then reveal whether $v$ belongs to $W$ by drawing an $p$-Bernoulli random variable. If $v \not\in W$, no edge connected to $v$ can be added to $M$, thus we mark $v$ as dead and discard $vu$. If $v \in W$, we add $vu$ to $M$ and then mark both $v$ and $u$ as dead as they cannot be matched anymore. One can confirm that at the end of this process, $M$ is precisely equivalent to matching $\lfmm{G[W \cup U], \pi}$.
	
	Let us fix an infinite tape of independent $p$-Bernoulli random variables $\vec{x} = (x_1, x_2, x_3, \ldots)$ and use it in the following way: Once we encounter the $i$'th potential-match edge $vu$ whose vertex $v$ is matched in matching $M$ (this is the matching in the statement of lemma, not to be confused with the greedy matching $M'$ we are constructing), we use the value of $x_i$ as the indicator of the event $v \in W$. Observe that if $x_i = 1$, then this edge $uv$ will be added to $M'$. Moreover, define $X := \sum_{i=1}^D x_i$ to be the number of 1's in the tape that we encounter by the end of process. (Here $D$ is the upper bound on the number of times that we reveal a random variable from the tape.) One can confirm that $X$ is precisely the number of vertices in the $V$-side of $M$ that are matched in $M'$: The precise quantity that lemma requires the lower bound for.
	
	Since each variable $x_i$ is 1 independently with probability $p$, we expect $X$ to be $pD$. However, note that the value of $D$ itself is a random variable depending on the randomizations revealed. Nonetheless, since $D$ is a stopping time for the process, we can use Wald's equation \cite{DBLP:books/daglib/0012859} to argue that the expected value of $X$ is indeed at least $p \E[D]$. Therefore, it suffices to show that $\E[D] \geq (|M|-2p|V|)$ to prove $\E[X] \geq p(|M|-2p|V|)$ as required by the lemma.
	
	To see why $\E[D]$ is this large, observe that for any vertex $v \in V$ that is matched in $M$ say via edge $vu$, we will encounter a potential-match edge unless $u$ is matched to another vertex. However, since each vertex in $V$ is sampled into $W$ with probability $p$, there are in expectation at most $p|V|$ vertices in $W$. Each such vertex can destroy at most two edges in $M$. For the rest of $|M|-2p|V|$ edges, we encounter at least a potential-match edge for which we reveal a random variable of the tape. Thus $\E[D] \geq |M|-p|V|$ and $\E[X] \geq p(|M|-2p|V|)$ as desired.
\end{proof}

\section{Missing Proofs}\label{sec:missingproofs}

\begin{proof}[Proof of Lemma~\ref{lem:updatedsafteredgeupdate}]
	Updating $k_i$ is easy. If for a vertex $v$, $k_i(v)$ has to be updated, then an edge incident to it must be in $L$. On the other hand, there are at most two edges connected to each vertex in $L$: At most one edge incident to it can join the matching and at most one was in it to leave. Therefore, by simply iterating over the edges in $L$, we can update $k_i(v)$ of any vertex necessary. This takes $O(|L|)$ time, which by Lemma~\ref{lem:greedyMM} is w.h.p. bounded by $O(\log n)$.
	
	Updating $N_i$ is more tricky. If an edge joins the matching, it can now become the eliminator of many other edges and if the eliminator of an edge $uv$ changes, we have to re-index $u$ and $v$ in each other's adjacency list $N_i(v)$. The crucial observation is that since the matching is constructed greedily, the matching on edges with rank in $[0, \pi_i(e))$ remains unchanged after inserting/deleting $e$. As a result, if the eliminator of an edge had rank in $[0, \pi_i(e))$, it remains to be its eliminator. Therefore, all the changes occur in the subgraph including edges with eliminator rank before the update was larger than $\pi_i(e)$. By Lemma~\ref{lem:sparsification} this graph has maximum degree $\min\{\Delta, O(\frac{\log n}{\pi_i(e)})\}$ w.h.p. Moreover, we can iterate over all neighbors $g$ of any edge $f \in L$ with $\elim_{G_i, \pi_i}(g) > \pi_i(e)$ in time $\min\{\Delta, \frac{1}{\pi_i(e)}\} \polylog n$. Re-indexing each also takes at most time $O(\log \Delta)$. Thus, the overall time required is
	\begin{flalign*}
	|L| \times \E\left[ \min\left\{\Delta, \frac{1}{\pi_i(e)}\right\} \polylog n \right] &\leq \E_{\pi_i(e) \sim [0, 1]}\left[ \min\left\{\Delta, \frac{1}{\pi_i(e)}\right\} \right] \polylog n  && \text{W.h.p. $|L| = O(\log n)$.} \\
	&= O(\log \Delta) \times \polylog n = \polylog n,
	\end{flalign*}
	completing the proof.
\end{proof}

\begin{proof}[Proof of Claim~\ref{cl:gkremainsunchanged}]
	Fix a graph $G_k$ with $k > j$. We first prove no edge is removed from $G_k$ after updating $f$. Take an edge $uv$ that is in $G_k$. By Algorithm~\ref{alg:meta}, one endpoint of this edge, say $v$ w.l.o.g., should be matched in $M_0$ via an edge that belongs to $S_k$. The other endpoint $u$, is either unmatched in $M_0$ or matched via an edge $ux$ where $ux \in S_\ell$ for some $\ell < k$. After the update, $v$ remains to be matched to the same vertex $u$ in $M_0$ for the following reason. Since $f \in S_j$, $uv \in S_k$, and $k > j$, then $\pi_0(f) > \pi_0(uv)$. As a result, since matching $M_0$ is constructed greedily by processing the edges in the increasing order of ranks, all the edges processed before $f$ that are in the matching will remain in the matching no matter if $f$ is in the graph or not, meaning that $uv$ will remain in $M_0$. Moreover, even though edge $ux$ may leave matching $M_0$, the edge to which $u$ will be matched  after the update (if any) will have rank at least $\pi_0(f)$. Thus $u$ will remain in set $U_k$ and as a result, the edge $vu$ will remain in $G_k$.
	
	A similar argument shows that any edge $uv$ in $G_k$ after the update, should have been in $G_k$ before the update too. More precisely, if $v$ is the part of edge $uv$ that is matched in $S_k$ after the update, then its matching edge in $M_0$ should have been in $M_0$ before the update too for precisely the same reason mentioned above. Moreover, no matter the update, if the other endpoint $u$ is in set $U_k$ after the update, it should have been in $U_k$ before the update too. Implying that $uv$ should have also been in $G_k$ before the update.

	Combination of the arguments of the two paragraphs above implies that graph $G_k$ remains exactly the same after and before the update.
\end{proof}


\end{document}